\documentclass{rendiconti}
\RendicontiPagina{1}{Mattia Scomparin}{Mattia Scomparin}

\usepackage{amsmath}
\usepackage{mathrsfs}

\usepackage{amsfonts}
\usepackage{slashed}
\usepackage{xcolor} 
\usepackage{amsmath,amsfonts,amssymb}

\usepackage[pagewise]{lineno}
\usepackage{tikz}
\usetikzlibrary{positioning}

\usepackage{outlines}

\title{Spectral separation of variables from equivalent Lagrangian systems}
\author{Mattia Scomparin}
\dedicated{mattia.scompa@gmail.com}

\summary{
We investigate the dynamical equivalence of quadratic Lagrangians and its
relation to separation of variables. We show that requiring two quadratic
Lagrangians to generate the same Euler--Lagrange equations imposes a
compatibility condition between the kinetic matrices and the potential.
For constant symmetric kinetic matrices, this condition reduces to a
commutation relation with the Hessian of the potential, yielding an
orthogonal spectral decomposition of the configuration space. The equations
of motion then decouple into independent subsystems: generically in
block-separated form, and completely when the spectrum is simple.
Applications include the Sawada--Kotera system and an $n$-dimensional
extension of the H\'{e}non--Heiles model, where the classical integrable
parameter regimes are recovered.
}

\intesta

\begin{document}

\allowbreak
\maketitle


\section{Introduction}
\label{sec:intro}

The separation of variables in Lagrangian systems is a central tool in the analysis of nonlinear equations of motion, allowing for the reduction of coupled dynamics to lower-dimensional subsystems and for the explicit construction of first integrals.

Classical approaches to separability are traditionally formulated within the
Hamilton--Jacobi framework. Additive separation theory originates in
St\"{a}ckel's classification of orthogonal coordinates~\cite{Stackel1893} and in
Levi-Civita's analytic separability conditions~\cite{LeviCivita1904}, which
provide necessary and sufficient differential criteria for separability in
given coordinates. Eisenhart subsequently reformulated the theory in
Riemannian terms, identifying Killing vectors, Killing tensors, and
separable webs as the geometric structures underlying separability for
natural Hamiltonians~\cite{Eisenhart1934}. In this setting, separability is
closely related to the existence of polynomial first integrals of low
degree~\cite{Woodhouse1975,KalninsMiller1980}. Benenti later developed an
intrinsic characterization of Hamilton--Jacobi separability in terms of
Killing tensors and commuting quadratic integrals~\cite{Benenti1997}, with
extensions to systems with scalar and vector
potentials~\cite{BenentiChanuRastelli2001}.

Several developments suggest a corresponding Lagrangian interpretation of
these constructions. Jacobi--Maupertuis reductions naturally transfer
Hamilton--Jacobi separation to tangent-bundle formulations, while null
Hamilton--Jacobi separation can be described through generalized
Levi-Civita conditions and Jacobi--Maupertuis-type
transformations~\cite{BenentiChanuRastelli2005}. Lagrangian approaches based
on adjoint symmetries and invariant volume forms provide constructive
procedures for obtaining separated coordinates and conserved
quantities~\cite{SarletRamos2000}. Related ideas also arise in the inverse
problem of the calculus of variations, where alternative Lagrangians are
characterized by multiplier matrices satisfying Helmholtz-type
conditions~\cite{PrinceSarletThompson2002}. Closely connected results appear
in the theory of geodesically equivalent metrics, where common
unparameterized geodesics lead to commuting integrals under suitable
spectral assumptions~\cite{TopalovMatveev2003}.

Motivated by these developments, we formulate a Lagrangian approach to
separation of variables based on the spectral structure of the kinetic term.
Rather than assuming separability of the potential, we consider dynamical
equivalence between quadratic Lagrangians with distinct constant symmetric
kinetic matrices. Requiring the corresponding Euler--Lagrange equations to
coincide yields an algebraic compatibility condition involving the Hessian
of the potential. This condition induces an orthogonal decomposition of the
configuration space into invariant eigenspaces, with respect to which the
equations of motion decouple into independent subsystems. Simple spectra
produce complete separation, whereas spectral degeneracies lead naturally to
block-separated dynamics.

The resulting framework provides an explicit mechanism for constructing
separated coordinates and quadratic first integrals directly at the
Lagrangian level. The equivalence between quadratic Lagrangians is therefore
not merely redundant, but imposes strong constraints on the admissible
potentials and on the structure of the dynamics itself.

The theory is illustrated through several examples, including the
Sawada--Kotera system and higher-dimensional extensions of the
H\'{e}non--Heiles model, where the spectral condition recovers the classical
integrable parameter regimes. We also address an inverse problem for systems
with non-symmetric constant kinetic matrices, characterizing those that are
dynamically equivalent to completely separated systems. In this setting,
separability is shown to be highly restrictive and to select quadratic
effective potentials.

To our knowledge, the equivalence between quadratic Lagrangians has not
previously been used as a direct mechanism for deriving separability from
algebraic compatibility conditions on the Hessian of the potential.

The paper is organized as follows.
Section~\ref{sec:mmm} derives the equivalence conditions for quadratic
Lagrangians and introduces the spectral separation framework.
Section~\ref{sec:sep} studies the resulting separability properties.
Section~\ref{sec:app} presents explicit examples, including the
Sawada--Kotera system, higher-dimensional H\'{e}non--Heiles models, and an
example in $\mathbb{R}^4$ with a transcendental potential.
Finally, Section~\ref{sec:inv} addresses the inverse characterization of
separable systems with non-symmetric kinetic matrices.

\textbf{Conventions.} Throughout the paper, the configuration space $\mathbb{R}^n$ is endowed with the standard Euclidean inner
product $u\cdot v\equiv\sum_{i=1}^nu_iv_i$ and the associated norm $\|w\|\equiv\sqrt{w\cdot w}$. A superposed dot denotes differentiation with respect to time, namely
$\dot{\phantom{x}}=d/dt$. Here and throughout the paper:
(i) $C^2(\mathbb{R}^n)$ denotes the space of twice continuously differentiable
functions $f:\mathbb{R}^n\to\mathbb{R}$;
(ii) $\mathrm{Sym}(n,\mathbb{R})$ denotes the space of real symmetric
$n\times n$ matrices;
(iii)
$O(n)=\{Q\in\mathbb{R}^{n\times n} : Q^T Q=\mathbb I_n\}$
denotes the orthogonal group.
We also use the Kronecker delta $\delta_{ij}$ and, unless otherwise specified,
summation indices are understood to run from $1$ to $n$.


\section{Quadratic Lagrangians and associated dynamics}
\label{sec:mmm}

Let $q(t)\in\mathbb{R}^n$ denote the configuration vector, and let $\dot q = dq/dt$.
We consider autonomous Lagrangians of the form $\mathcal L(q,\dot q)$, whose
dynamics is governed by the Euler--Lagrange equations
\begin{equation}
\label{eq:euler-lagrange}
\big(\partial_{\dot q_i}\mathcal L\big)^{\cdot}
-
\partial_{q_i}\mathcal L
=
0,
\qquad
i=1,\dots,n.
\end{equation}
Here, $\partial_{q_i}$ and $\partial_{\dot q_i}$ denote partial derivatives with
respect to $q_i$ and $\dot q_i$, respectively.
Along any solution of~\eqref{eq:euler-lagrange}, the associated energy
\begin{equation}
\label{eq:energy}
E(q,\dot q)
\equiv
\partial_{\dot q_i}\mathcal L\,\dot q_i-\mathcal L\,,
\end{equation}
is conserved.

We shall focus on the following two smooth Lagrangians:
\begin{equation}
\label{eq:LtL}
\mathcal{L}(q,\dot q)
=
\frac{1}{2} \sum_{i=1}^n \dot q_i^2 - U(q),
\qquad
\widetilde {\mathcal{L}}(q,\dot q)
=
\frac{1}{2} \sum_{i,j=1}^n A_{ij} \dot q_i \dot q_j - \widetilde U(q),
\end{equation}
where the potentials $U,\widetilde U\in C^2(\mathbb{R}^n)$ and
$A\in\mathrm{Sym}(n,\mathbb{R})$ is a constant, invertible, symmetric matrix. 
This last condition is essential to our study and is not meant to address the
general case of $q$-dependent kinetic matrices. In addition,  Lagrangians \eqref{eq:LtL} are referred to as \emph{quadratic}, since their kinetic term depends quadratically on the velocity $\dot q$.

The Euler--Lagrange equations~\eqref{eq:euler-lagrange} associated with the
Lagrangians~\eqref{eq:LtL} are
\begin{equation}
\label{eq:LtLequations}
\ddot q_i + \partial_{q_i} U(q) = 0,
\qquad
\sum_{j=1}^n A_{ij}\,\ddot q_j + \partial_{q_i} \widetilde U(q) = 0,
\quad
i=1,\dots,n.
\end{equation}
The corresponding energy integrals~\eqref{eq:energy} take the form
\begin{equation}
\label{eq:LtLenergy}
E(q,\dot q)
=
\frac{1}{2} \sum_{i=1}^n \dot q_i^2 + U(q),
\qquad
\widetilde E(q,\dot q)
=
\frac{1}{2} \sum_{i,j=1}^n A_{ij} \dot q_i \dot q_j + \widetilde U(q).
\end{equation}

\begin{remark}[Trivial case]
If $A=\mathbb I_n$, i.e. $A_{ij}=\delta_{ij}$, then $\widetilde{\mathcal L}$ coincides with
$\mathcal L$ up to the replacement of the potential $U$ by $\widetilde U$.
In this case the Euler-Lagrange equations trivially coincide.
\end{remark}


\subsection{Equivalence of quadratic Lagrangian systems}
\label{sec:nnn}

We now address the question of when the two quadratic Lagrangians
\eqref{eq:LtL} generate the same dynamics.
\begin{proposition}
The Lagrangians \eqref{eq:LtL}
generate the same Euler--Lagrange equations \eqref{eq:euler-lagrange} if and only if
\begin{equation}
\label{eq:sameEL}
\partial_{q_i} \widetilde U(q) = \sum_{j=1}^n A_{ij} \, \partial_{q_j} U(q), \quad \forall i=1,\dots,n.
\end{equation}
Moreover, whenever \eqref{eq:sameEL} holds, the two energy functions in \eqref{eq:LtLenergy}
are both first integrals of the common dynamics.
\end{proposition}
\begin{proof}
Consider the two Euler--Lagrange equations \eqref{eq:LtLequations}. In order for them to coincide,
it is necessary and sufficient to cast the first equation into the second, which yields
condition \eqref{eq:sameEL}.
With reference to the energies, we show that $\widetilde E$ is conserved along the flow of
$\mathcal{L}$. Therefore, differentiating $\widetilde E$ and evaluating along the equations of motion
\eqref{eq:LtLequations} $\ddot q_j=-\partial_{q_j}U$ generated by $\mathcal{L}$ yields
\begin{equation}
\label{eq:passrr}
\textstyle
\dot{\widetilde E} =
\sum_{i=1}^n\dot q_i
\big(\sum_{j=1}^n A_{ij}\ddot q_j
+\partial_{q_i}\widetilde U
\big)=\sum_{i=1}^n
\dot q_i\big(
-
\sum_{j=1}^n
A_{ij}\partial_{q_j}U
+
\partial_{q_i}\widetilde U\big).
\end{equation}
Hence, $\widetilde E$ is conserved along the flow of $\mathcal{L}$, i.e.\ $\dot{\widetilde E} =0$,
provided that condition \eqref{eq:sameEL} holds.
The same reasoning applies to $E$ along the flow of $\widetilde{\mathcal{L}}$. Indeed,
differentiating $E$ gives $\dot E=\sum_{i=1}^n\dot q_i(\ddot q_i+\partial_{q_i}U)$; using
the equations of motion of $\widetilde{\mathcal{L}}$, namely
$\sum_{j=1}^n A_{ij}\ddot q_j + \partial_{q_i}\widetilde U=0$, together with \eqref{eq:sameEL},
one obtains $\ddot q_i+\partial_{q_i}U=0$ and therefore $\dot E=0$ along the flow of
$\widetilde{\mathcal{L}}$.
\end{proof}
Condition \eqref{eq:sameEL} can be interpreted as a compatibility requirement
between the matrix $A$ and the properties encoded by the potential $U$.
This interpretation is made precise by the following result.
\begin{proposition}
\label{eq:commHess}
Condition \eqref{eq:sameEL} is equivalent
to the requirement that the matrix $A$ commutes with the Hessian $\partial^2_{qq}U$ of $U$, namely
\begin{equation}
\big[\partial^2_{qq}U,A\big]_{ki}\equiv\sum_{j=1}^n\left( \partial^2_{q_k q_j}\!U A_{ji}-A_{ij} \partial^2_{q_j q_k} U\right)=0, \quad \forall i,k=1,\dots,n.
\end{equation}
\end{proposition}
\begin{proof}
Differentiating \eqref{eq:sameEL} with respect to $q_k$ yields
\begin{equation}
\textstyle
\partial^2_{q_k q_i} \widetilde U
=
\sum_{j=1}^n A_{ij}\, \partial^2_{q_k q_j} U,
\end{equation}
since $A$ is constant.
Since $\partial^2_{q_k q_i} \widetilde U$ is symmetric in the indices $(i,k)$,
the right-hand side must be symmetric as well, and therefore
\begin{equation}
\textstyle
\sum_{j=1}^n A_{ij}\, \partial^2_{q_k q_j} U
=
\big(\sum_{j=1}^n A_{ij}\, \partial^2_{q_k q_j} U\big)^T
=
\sum_{j=1}^n (\partial^2_{q_k q_j} U)^T A_{ij}^T.
\end{equation}
Using the symmetry of both $\partial^2_{qq}U$ and $A$, this becomes
\begin{equation}
\textstyle
\sum_{j=1}^n (\partial^2_{q_k q_j} U)^T A_{ij}^T
=
\sum_{j=1}^n \partial^2_{q_j q_k} U\, A_{ij}.
\end{equation}
Relabeling the summation indices finally yields
\begin{equation}
\textstyle
\sum_{j=1}^n A_{ij}\, \partial^2_{q_j q_k} U
=
\sum_{j=1}^n \partial^2_{q_k q_j} U\, A_{ji},
\end{equation}
which is equivalent to the commutation condition
$\,[\partial^2_{qq}U, A]=0$.
Conversely, if $U\in C^2(\Omega)$ with $\Omega\subseteq\mathbb{R}^n$ simply connected and $[D^2U(q),A]=0$ for all $q\in\Omega$, then for the vector field $F:=A\nabla U$ one has $\partial_{q_k}F_i=\sum_{j}A_{ij}\partial^2_{q_k q_j}U=\sum_{j}\partial^2_{q_i q_j}U\,A_{jk}=\partial_{q_i}F_k$ for all $i,k$, hence $F$ is conservative and there exists $\widetilde U\in C^2(\Omega)$ such that $\nabla\widetilde U=F=A\nabla U$, i.e.\ \eqref{eq:sameEL} holds.
\end{proof}
We now exploit the symmetry of $A$ to simplify the structure of the system.
Since $A \in \mathrm{Sym}(n,\mathbb{R})$, by the Spectral Theorem its eigenvalues
are real and there exists an orthonormal basis of $\mathbb{R}^n$ formed by
eigenvectors of $A$.

Let $Q$ be the matrix whose columns are such orthonormal eigenvectors.

\begin{remark}
In this case $Q \in O(n)$ because its columns form an
orthonormal basis, hence $Q^T Q=\mathbb I_n$.
Therefore,
\begin{equation}
\label{eq:diag}
Q^T A Q = D,
\qquad
D = \mathrm{diag}
(\underbrace{\lambda_1,\dots,\lambda_1}_{m_1},\dots,
\underbrace{\lambda_r,\dots,\lambda_r}_{m_r}),
\end{equation}
where $\lambda_1,\dots,\lambda_r$ are the distinct eigenvalues of $A$
and $m_1,\dots,m_r$ their algebraic multiplicities.
\end{remark}

This diagonalization naturally suggests introducing adapted coordinates.

\begin{definition}[Spectral coordinates]
We introduce the new spectral coordinates $x=Q^T q$, whose components are
\begin{equation}
\label{eq:xCoord}
x_i
=
\sum_{j=1}^n Q^T_{ij} q_j,
\qquad
i=1,\dots,n.
\end{equation}
\end{definition}
In what follows, we adopt the notation $U(x)=U(q)\vert_{q=Qx}$, and similarly for
$\widetilde U$. The following remarks collect several basic identities that will be repeatedly used in the subsequent computations.
\begin{remark}
\label{rem:calc1a}
In components, since 
\begin{equation}
\textstyle
q_i=\sum_{j=1}^n(Q^T)^{-1}_{ij} x_j=\sum_{j=1}^nQ_{ij} x_j,
\end{equation}
we have that 
\begin{equation}
\textstyle
\dot q_i=\sum_{j=1}^nQ_{ij} \dot x_j,
\quad\mbox{and}\quad
\ddot q_i=\sum_{j=1}^nQ_{ij} \ddot x_j.
\end{equation}
\end{remark}
\begin{remark}
\label{rem:calc2a}
Recalling that $Q$ is not symmetric, we compute
\begin{equation}
\textstyle
\partial_{q_i} x_j
=
\partial_{q_i}\sum_{k} Q^T_{jk} q_k
=
\sum_{k} Q^T_{jk}\, \partial_{q_i} q_k
=
\sum_{k} Q^T_{jk}\, \delta_{ki}
=
Q^T_{ji}.
\end{equation}
Hence,
\begin{equation}
\textstyle
\partial_{q_i}
=
\sum_j (\partial_{q_i} x_j)\, \partial_{x_j}
=
\sum_j Q^T_{ji}\, \partial_{x_j}
=
\sum_j Q_{ij}\, \partial_{x_j}.
\end{equation}
\end{remark}
We can now express the equations of motion in spectral coordinates.
\begin{proposition}
\label{eq:xEL}
 In terms of the spectral coordinates \eqref{eq:xCoord}, the Euler--Lagrange equations \eqref{eq:LtLequations} take the form
\begin{equation}
\label{eq:xEoM}
\ddot{x}_k + \partial_{x_k}U(x) = 0,
\qquad
\lambda_k \ddot{x}_k + \partial_{x_k} \widetilde U(x) = 0,
\qquad k = 1,\dots,n.
\end{equation}
\end{proposition}
\begin{proof}
Consider the relations collected in Remark~\ref{rem:calc1a}.
Starting from the first equation in~\eqref{eq:LtLequations}, we compute
\begin{equation}
0=\textstyle\ddot q_i + \partial_{q_i} U=\sum_{j}Q_{ij} \ddot x_j+\sum_j Q_{ij}\partial_{x_j}U=\sum_{j}Q_{ij} \big(\ddot x_j+\partial_{x_j}U\big).
\end{equation}
Since $Q$ is invertible, the resulting relation implies
$\ddot x_j+\partial_{x_j}U$ for all $j$, which gives the
first equation in~\eqref{eq:xEoM}.

We now turn to the second equation in~\eqref{eq:LtLequations}.
Using again Remark~\ref{rem:calc1a} with Remark~\ref{rem:calc2a}, we obtain
\begin{equation}
\label{eq:pass}
\textstyle
\begin{aligned}
\textstyle
0 &=\textstyle  \sum_{j} A_{ij} \ddot q_j + \partial_{q_i} \widetilde U,\\
&= \textstyle\sum_{j} A_{ij} \sum_{k}Q_{jk} \ddot x_k + \sum_mQ_{im}\partial_{x_m} \widetilde U,\\
&=\textstyle\sum_{j} \big(\sum_p \delta_{ip} A_{pj}\big) \sum_{k}Q_{jk} \ddot x_k + \sum_mQ_{im}\partial_{x_m} \widetilde U,\\
&=\textstyle\sum_{j} \sum_p \big(\sum_mQ_{im}Q^T_{mp}\big) A_{pj} \sum_{k}Q_{jk} \ddot x_k + \sum_mQ_{im}\partial_{x_m} \widetilde U,\\
&= \textstyle \sum_mQ_{im}\big[\sum_k\big(\sum_{j}\sum_pQ^T_{mp}A_{pj}Q_{jk}\big)\ddot x_k+\partial_{x_m} \widetilde U\big],\\
&= \textstyle \sum_mQ_{im}\big(\sum_k\lambda_k \delta_{mk}\ddot x_k+\partial_{x_m} \widetilde U\big),\\
&= \textstyle \sum_mQ_{im}\big(\lambda_m \ddot x_m+\partial_{x_m} \widetilde U\big), \qquad i=1,\dots,n.
\end{aligned}
\end{equation}
In the above computation, the identity
$\delta_{ip}=\sum_m Q_{im}Q^T_{mp}$ is used to factor out the matrix $Q$ and
rearrange the summations. The diagonalization $Q^T A Q=\mathrm{diag}(\lambda_1,
\dots,\lambda_n)$ then yields
$\sum_{j,p} Q^T_{mp}A_{pj}Q_{jk}=\lambda_k\delta_{mk}$, leading to the last line.
Since $Q$ is invertible, the resulting relation of \eqref{eq:pass} implies
$\lambda_m \ddot x_m+\partial_{x_m}\widetilde U=0$ for all $m$, which gives the
second equation in~\eqref{eq:xEoM}.
\end{proof}

\begin{remark} Consider again the relations obtained in Remark \ref{rem:calc1a} and Remark \ref{rem:calc2a}. In terms of the new coordinates, the corresponding (conserved) energy integrals \eqref{eq:LtLenergy} are
\begin{equation}
\label{eq:energies}
E(x,\dot{x})=\frac{1}{2} \sum_{k=1}^n \dot{x}_k^2+U(x),
\quad\quad
E(x,\dot{x})=\frac{1}{2} \sum_{k=1}^n \lambda_k \dot{x}_k^2+\widetilde U(x)\,.
\end{equation}
\end{remark}
\begin{proposition}
In terms of spectral coordinates, the Euler--Lagrange equations \eqref{eq:xEL} coincide if
\begin{equation}
\label{eq:xsameEL}
\partial_{x_i} \widetilde U
=
\lambda_i\, \partial_{x_i} U,
\qquad i=1,\dots,n.
\end{equation}
\end{proposition}

\begin{proof}
Multiplying both sides of \eqref{eq:sameEL} by $Q^T$ yields
\begin{equation}
\textstyle
\sum_{k}Q^T_{ki}\partial_{q_i} \widetilde U = \sum_{k} \sum_{j}Q^T_{ki}A_{ij} \, \partial_{q_j} U=\sum_{k} \sum_{j}Q^T_{ki}A_{ij} \, \partial_{q_j} U.
\end{equation}
Using Remark~\ref{rem:calc2a} to express derivatives with respect to $q$ in terms
of derivatives with respect to $x$, we obtain
\begin{equation}
\textstyle
\sum_q
\Big(\sum_i Q^T_{ki} Q_{iq}\Big)\partial_{x_q} \widetilde U
=
\sum_q
\Big(\sum_i \sum_j Q^T_{ki} A_{ij} Q_{jq}\Big)\partial_{x_q} U.
\end{equation}
Since
$\sum_i Q^T_{ki} Q_{iq}=\delta_{kq}$ and
$\sum_i \sum_j Q^T_{ki} A_{ij} Q_{jq}=\lambda_q\,\delta_{kq}$,
it follows that
\begin{equation}
\textstyle
\sum_q \delta_{kq}\,\partial_{x_q} \widetilde U
=
\sum_q \lambda_q\,\delta_{kq}\,\partial_{x_q} U.
\end{equation}
Removing the Kronecker delta yields \eqref{eq:xsameEL}.
\end{proof}

\begin{remark}
\label{rem:sameEoM}
As a simple check, one can verify directly that imposing
\eqref{eq:xsameEL} into \eqref{eq:xEoM} makes the two systems coincide. Indeed,
$
\textstyle
0
=
\lambda_k \ddot x_k + \partial_{x_k} \widetilde U
=
\lambda_k \ddot x_k + \lambda_k \partial_{x_k} U
=
\lambda_k \big(\ddot x_k + \partial_{x_k} U\big),
$
and therefore
$\ddot x_k + \partial_{x_k} U = 0$, for all $k=1,\dots,n$.
\end{remark}

\begin{remark}[Gauge transformation]
Given a Lagrangian $\mathcal L(q,\dot q)$ on $\mathbb R^n$ and a smooth scalar function $G(q)$, it is well known that
one can construct infinitely many Lagrangians $
\mathcal L'(q,\dot q)
=
\mathcal L(q,\dot q) + \nabla G(q)\cdot \dot q
$
generating the same
Euler--Lagrange equations of  $\mathcal L$.
In this case, $\mathcal L$ and
$\mathcal L'$ are said to be related by a \emph{gauge transformation}.
Our two quadratic Lagrangians \eqref{eq:LtL} are not related by a gauge
transformation, since their difference is quadratic in the velocities rather
than linear as $\nabla G(q)\cdot \dot q$. Therefore, their equivalence  stated in Proposition~\ref{theo:eq_statmnt} cannot be attributed to gauge
freedom.
\end{remark}

The above construction allows us to recast the problem of separability in terms of spectral coordinates, which will be crucial in what follows.


\section{Separability of equivalent quadratic Lagrangian systems}
\label{sec:sep}

In this section we investigate the separability properties of the quadratic
Lagrangian systems introduced above. 
Rather than assuming separability of the
potential \emph{a priori}, our approach is based on compatibility of Euler--Lagrange equations.
Indeed, the equivalence condition \eqref{eq:commHess} constrains the potential in a way that is
most naturally expressed in spectral coordinates. 

We begin by fixing a spectral setting adapted to the symmetric matrix $A$,
which will be used throughout the section to describe both the coordinates
and the structure of the potentials.
\begin{definition}[Spectral framework]
\label{def:spectral}
Let $A \in \mathrm{Sym}(n,\mathbb{R})$. We adopt the following spectral framework:
\begin{enumerate}
\item Let $Q \in O(n)$ diagonalize $A$, that is
$
Q^T A Q = \mathrm{diag}(\lambda_1 I_{m_1}, \dots, \lambda_r I_{m_r}),
$
where $\lambda_1,\dots,\lambda_r$ are the distinct eigenvalues of $A$ with
multiplicities $m_1,\dots,m_r$.
\item Decompose $Q = [Q_{(1)}, \dots, Q_{(r)}]$, where the columns of
$Q_{(k)} \in \mathbb{R}^{n \times m_k}$ form an orthonormal basis of the eigenspace
$\mathcal{E}_{\lambda_k}$.
\item Define the spectral coordinates $x := Q^T q \in \mathbb{R}^n$, with
$
x_{(k)} := Q_{(k)}^T q \in \mathbb{R}^{m_k},
$
so that $x = (x_{(1)}, \dots, x_{(r)})^T$.
\end{enumerate}
\end{definition}
Let us restate, in spectral coordinates, the equivalence condition previously
derived between the two quadratic Lagrangian systems.
\begin{proposition}[Equivalence conditions]
\label{theo:eq_statmnt}
Assume the spectral framework of Definition~\ref{def:spectral}. The following
statements are equivalent:
\begin{enumerate}
\item
The Lagrangians
$
\mathcal{L}(q,\dot q)=\tfrac12|\dot q|^2-U(q)
$
and
$
\widetilde{\mathcal{L}}(q,\dot q)=\tfrac12 \dot q^T A \dot q-\widetilde U(q)
$
generate the same Euler--Lagrange equations.
\item
The potentials satisfy
$
\partial_{x_i}\widetilde U=\lambda_i\,\partial_{x_i}U$, with $i=1,\dots,n.
$
\end{enumerate}
\end{proposition}
\begin{proof}
The implication $(1)\Rightarrow(2)$ follows from
\eqref{eq:xsameEL} in Section~\ref{sec:nnn}.
The converse $(2)\Rightarrow(1)$ follows from Remark~\ref{rem:sameEoM}.
\end{proof}
We now show that the equivalence condition has strong structural consequences
on the form of the potentials, forcing a block-wise separation dictated by
the spectral decomposition of $A$.
\begin{theorem}[Weak spectral separation and first integrals]
\label{th:weak}
Under the assumptions of Proposition~\ref{theo:eq_statmnt}, the potentials split
along the spectral decomposition of $A$ as
\begin{equation}
\label{eq:weak-potentials}
U=\sum_{k=1}^r H_k[x_{(k)}],
\qquad
\widetilde U=\sum_{k=1}^r \lambda_k H_k[x_{(k)}]+C,
\end{equation}
for suitable smooth functions $H_k:\mathbb{R}^{m_k}\to\mathbb{R}$. Consequently,
the Euler--Lagrange equations decouple into $r$ independent blocks,
\begin{equation}
\label{eq:weak-eom}
\ddot x_{(k)}+\nabla H_k[x_{(k)}]=0,
\qquad k=1,\dots,r.
\end{equation}
Moreover, the corresponding block energy sums
\begin{equation}
\label{eq:weak-energies}
E=\sum_{k=1}^r\!\Big(\tfrac12|\dot x_{(k)}|^2+H_k[x_{(k)}]\Big),
\quad
\widetilde E=\sum_{k=1}^r \lambda_k\!\Big(\tfrac12|\dot x_{(k)}|^2+H_k[x_{(k)}]\Big)+C,
\end{equation}
are conserved along solutions.
\end{theorem}
\begin{proof}
$[$Step $1]$
Let $x_\alpha \in (\lambda_a, \mathcal{E}_{\lambda_a})$ and $x_\beta \in (\lambda_b, \mathcal{E}_{\lambda_b})$ be spectral coordinates associated with two distinct block eigenspaces ($\lambda_a \neq \lambda_b$). Using assumption (2) of Proposition \ref{theo:eq_statmnt}, differentiate $\partial_{x_\alpha} \widetilde U = \lambda_a \partial_{x_\alpha} U$ with respect to $x_\beta$ to obtain
$
\partial^2_{x_\beta x_\alpha} \widetilde U = \lambda_a \partial^2_{x_\beta x_\alpha} U.
$
Similarly, differentiating $\partial_{x_\beta} \widetilde U = \lambda_b \partial_{x_\beta} U$ with respect to $x_\alpha$ yields
$
\partial^2_{x_\alpha x_\beta} \widetilde U = \lambda_b \partial^2_{x_\alpha x_\beta} U.
$
Since $U, \widetilde U \in C^2(\mathbb{R}^n)$, mixed partial derivatives commute, hence
$
\partial^2_{x_\alpha x_\beta} \widetilde U = \partial^2_{x_\beta x_\alpha} \widetilde U.
$
Therefore,
$
\lambda_a \partial^2_{x_\alpha x_\beta} U = \lambda_b \partial^2_{x_\alpha x_\beta} U,
$
which implies
$
(\lambda_a - \lambda_b)\, \partial^2_{x_\alpha x_\beta} U = 0.
$
Since $\lambda_a \neq \lambda_b$, it follows that
$
\partial^2_{x_\alpha x_\beta} U = 0.
$
By the same argument, one also obtains $\partial^2_{x_\alpha x_\beta} \widetilde U = 0$. Thus, all mixed derivatives between coordinates belonging to distinct eigenspaces vanish.
$[$Step $2]$
Since $\partial^2_{x_\alpha x_\beta} U = 0$ and $\partial^2_{x_\alpha x_\beta} \widetilde U = 0$, integration of these conditions implies the existence of smooth functions
$
H_k : \mathbb{R}^{m_k} \to \mathbb{R}$ and 
$\widetilde H_k : \mathbb{R}^{m_k} \to \mathbb{R},
$
such that
$
U(x) = H_1[x_{(1)}] + \cdots + H_r[x_{(r)}]$ and 
$\widetilde U(x) = \widetilde H_1[x_{(1)}] + \cdots + \widetilde H_r[x_{(r)}].
$
Using assumption (2) of Proposition \ref{theo:eq_statmnt}, namely $\partial_{x_\alpha} \widetilde U = \lambda_a \partial_{x_\alpha} U$, we deduce that
$
\widetilde H_k[x_{(k)}] = \lambda_k\, H_k[x_{(k)}].
$
Substituting these expressions into Eqs.~\eqref{eq:xEoM} and \eqref{eq:energies}, the desired results follow.
\end{proof}
A particularly relevant situation occurs when all the eigenspaces of $A$
are one-dimensional, leading to a complete separation of the dynamics.
\begin{theorem}[Full spectral separation and first integrals]
\label{th:strong}
Assume that the equivalent conditions of Theorem~\ref{theo:eq_statmnt} hold and
that each eigenspace of $A$ is one-dimensional, i.e.\ $\dim \mathcal{E}_{\lambda_k}=m_k=1$
for all $k$. Then the potentials $U$ and $\widetilde U$ are fully separable and
can be written as
\begin{equation}
\label{eq:strong-potentials}
U = \sum_{k=1}^n f_k(x_k),
\qquad
\widetilde U = \sum_{k=1}^n \lambda_k f_k(x_k),
\end{equation}
for suitable smooth functions $f_k:\mathbb{R}\to\mathbb{R}$. Consequently, the
Euler--Lagrange equations decouple completely,
\begin{equation}
\label{eq:strong-eom}
\ddot x_k + f_k'(x_k) = 0,
\qquad k=1,\dots,n,
\end{equation}
and each coordinate admits an associated conserved energy. In particular, the
total conserved quantities take the form
\begin{equation}
\label{eq:strong-energies}
E = \sum_{k=1}^n\!\Big(\tfrac12 \dot x_k^2 + f_k(x_k)\Big),
\qquad
\widetilde E = \sum_{k=1}^n \lambda_k\Big(\tfrac12 \dot x_k^2 + f_k(x_k)\Big).
\end{equation}
\end{theorem}
\begin{proof}
If each spectral block has dimension $m_k=1$, then $x_{(k)}$ reduces to the single
coordinate $x_k$, and each function
$H_k:\mathbb{R}^{m_k}\to\mathbb{R}$ appearing in
\eqref{eq:weak-potentials} becomes a scalar function
$f_k:\mathbb{R}\to\mathbb{R}$. Accordingly, the block-decomposed expressions for
the potentials, the equations of motion, and the conserved quantities given in
\eqref{eq:weak-potentials}, \eqref{eq:weak-eom}, and \eqref{eq:weak-energies}
reduce exactly to
\eqref{eq:strong-potentials}, \eqref{eq:strong-eom}, and
\eqref{eq:strong-energies}, respectively.
\end{proof}
\begin{remark}
The condition $m_k=1$ for all $k$ is equivalent to the absence of nontrivial
orthogonal symmetries within the eigenspaces of $A$.
\end{remark}
\begin{remark}[Failure case]
When $A$ depends on $q$, full spectral separation generally breaks down, as
exemplified by Hamiltonian systems on Riemannian manifolds; this case will be
studied in Section~\ref{sec:HHhigher} through the $n$-dimensional H\'enon--Heiles
system.
\end{remark}


\subsection{Overview of the main result}

The geometric mechanism underlying our construction is schematically illustrated in Figure~\ref{fig:block-separation}. Requiring two quadratic Lagrangians to produce the same equations of motion imposes a spectral decoupling of the dynamics. More precisely, the system decomposes into independent blocks associated with the eigenspaces of the kinetic matrix. Complete separability occurs exactly when all eigenspaces are one-dimensional, yielding dynamically independent degrees of freedom, each endowed with its own conserved energy. Notably, no separability ansatz or Killing tensor structure is assumed \emph{a priori}.

\begin{figure}[h]
\centering
\begin{tikzpicture}[
    >=stealth,
    node distance=5mm and 9mm,
    every node/.style={align=center, font=\scriptsize}
]

\node (L)
{$
\begin{array}{c}
\mathcal{L} = \tfrac12 \dot q^2 - U(q) \\[1mm]
\widetilde{\mathcal{L}} = \tfrac12 \dot q^T A \dot q - \widetilde U(q) \\[1mm]
\mbox{\footnotesize Same Euler--Lagrange equations}
\end{array}
$};

\node (cond) [below=of L]
{$
\begin{array}{c}
\big[\partial^2_{qq}U,A\big]=0_{n\times n}\\[1mm]
\mbox{\footnotesize Compatibility condition}
\end{array}
$};

\node (diag) [below=of cond]
{
$
\begin{array}{c}
Q^T A Q =
\mathrm{diag}(\lambda_1 I_{m_1},\dots,\lambda_r I_{m_r}) \\[1mm]
\mbox{\footnotesize Spectral block-decomposition}
\end{array}
$
}
;

\node (E1) [below left=of diag]
{$x^{(1)}\in \mathcal{E}_{\lambda_1}$};

\node (Edots) [below=of diag]
{$\vdots$};

\node (Er) [below right=of diag]
{$x^{(r)}\in \mathcal{E}_{\lambda_r}$};

\node (H1) [below=of E1]
{$\begin{aligned}
U &= H_1[x^{(1)}] + \cdots \\
\widetilde U &= \lambda_1 \widetilde H_1[x^{(1)}] + \cdots
\end{aligned}$};

\node (Hr) [below=of Er]
{$\begin{aligned}
U &= \cdots + H_r[x^{(r)}] \\
\widetilde U &= \cdots + \lambda_r \widetilde H_r[x^{(r)}]
\end{aligned}$};

\node (EL) [below=8mm of Edots]
{
$
\begin{array}{c}
\mbox{\footnotesize Block-separated dynamics}\\[0.5mm]
(k=1,...,r)\\[1mm]
\ddot x_{(k)} + \partial H_k[x_{(k)}] = 0\\[1mm]
\lambda_k \ddot x_{(k)} + \partial \widetilde H_k[x_{(k)}] = 0
\end{array}
$

};

\draw[->] (L) -- (cond);
\draw[->] (cond) -- (diag);

\draw[->] (diag) -- (E1);
\draw[->] (diag) -- (Edots);
\draw[->] (diag) -- (Er);

\draw[->] (E1) -- (H1);
\draw[->] (Er) -- (Hr);

\draw[->] (Edots) -- (EL);

\end{tikzpicture}
\caption{\small
Geometric mechanism behind the theorems:
spectral decomposition of the constant symmetric matrix $A$
induces an orthogonal splitting
$\mathbb{R}^n=\bigoplus_k \mathcal{E}_{\lambda_k}$,
which forces block separation of the potentials and hence of the Euler--Lagrange equations.
}
\label{fig:block-separation}
\end{figure}
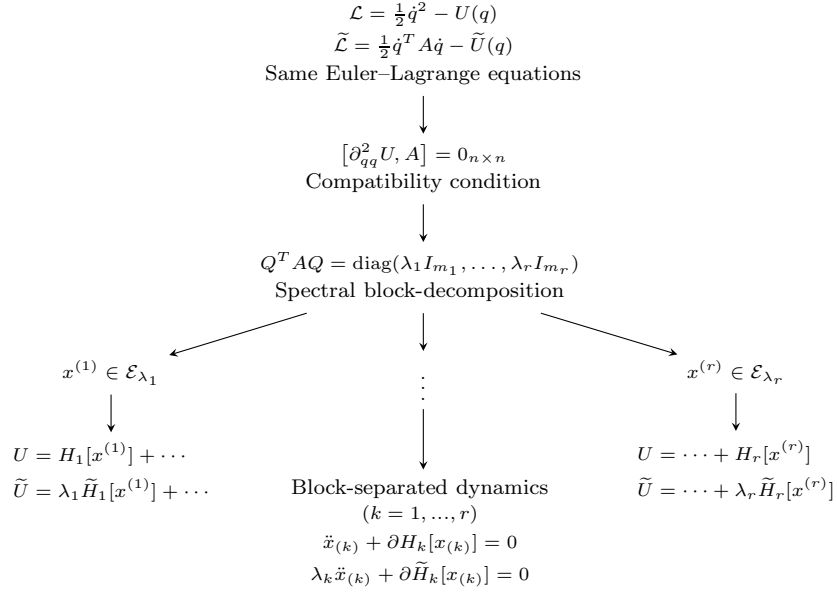

Our construction is closely related to the classical theory of Hamilton--Jacobi separability for natural Hamiltonian systems, particularly to St\"{a}ckel systems and the Benenti framework based on Killing tensors. In the classical St\"{a}ckel--Benenti setting, separability is formulated at the Hamiltonian level and relies on specific geometric assumptions: the kinetic energy determines a Riemannian metric, while the potential must satisfy compatibility conditions induced by an underlying St\"{a}ckel structure. The separation coordinates then arise from the joint spectral decomposition of the metric and the associated Killing tensors.

The present approach should therefore be viewed as complementary to the classical St\"{a}ckel--Benenti theory rather than as an alternative to it. Indeed, several systems that are separable in the classical sense also fit naturally within our framework, although the conceptual starting point is fundamentally different.


\section{Applications}
\label{sec:app}

We now describe how the spectral separation can be used in practice to analyze the separability and integrability properties of a given Lagrangian system. 

Consider a natural Lagrangian of the form
$
\mathcal L(q,\dot q)=\tfrac12\sum_{i=1}^n \dot q_i^2-U(q),
$
with a prescribed potential $U\in C^2(\mathbb R^n)$. The analysis proceeds as follows:
\begin{outline}[itemize]
\1 Compute the Hessian matrix $\partial^2_{qq}U$ of the potential. 
\1 Look for constant symmetric
matrices $A\in\mathrm{Sym}(n,\mathbb R)$ that commute with the Hessian.
In practice, writing $A=(a_{ij})$ with unknown coefficients and imposing the
commutation conditon \eqref{eq:commHess}  yields a linear system for the $a_{ij}$,
whose solutions determine the admissible expressions for $A$.
\1 Once the $A$ matrix is found, its spectral properties determine the structure
of the dynamics:
(i) If $A$ has distinct eigenvalues, the equations of motion split
into independent spectral blocks, leading to block separation (Theorem \ref{th:weak}). 
(ii)  If, moreover,
the spectrum of $A$ is simple, full separation of variables is achieved (Theorem \ref{th:strong}). 
\1 In both (i) and (ii) cases, the theory provides explicit conserved quantities, namely the
canonical energy $E$ and the additional integral $\widetilde E$ [Eqs. \eqref{eq:weak-energies} for (i) and Eqs. \eqref{eq:strong-energies} for (ii)].
\end{outline}
The following examples illustrate this procedure and show how classical
integrable systems fit naturally within the spectral separation framework.


\subsection{The $\mathbb{R}^2$ Sawada--Kotera case}

We consider the two--dimensional Sawada--Kotera system, arising as a
finite--dimensional reduction of the integrable Sawada--Kotera equation
introduced in soliton theory as a higher--order member of the
Korteweg--de Vries hierarchy~\cite{SawadaKotera1974}. In its infinite--dimensional
setting, the Sawada--Kotera equation admits a Lax representation, a
bi--Hamiltonian structure, and an infinite sequence of conserved quantities \cite{FuchssteinerOevel1982}.

A distinguished reduction to two degrees of freedom leads to a natural
mechanical system with a polynomial potential, endowed with additional first
integrals quadratic in the velocities and admitting separation of variables \cite{Fordy1983,CosgroveScoufis1993}.

Specifically, for $q=(q_1,q_2)\in\mathbb R^2$, the associated Lagrangian is
\begin{equation}
\mathcal{L}(q,\dot q)=\tfrac12 |\dot q|^2-U(q),
\qquad
U(q)=\tfrac12(q_1^2+q_2^2)+q_1^2 q_2+\tfrac13 q_2^3 .
\end{equation}
From the viewpoint of finite--dimensional dynamics, this system belongs to the
class of integrable cubic models related to the generalized
H\'enon--Heiles family, which plays a central role in the study of integrability
and the transition to chaotic behavior~\cite{HenonHeiles1964}.

We consider a generic (constant) symmetric matrix
\begin{equation}
A=
\begin{pmatrix}
a & b\\
b & d
\end{pmatrix}.
\end{equation}
The Hessian of the potential reads
\begin{equation}
\partial^2_{qq}U=
\begin{pmatrix}
1+2q_2 & 2q_1 \\
2q_1 & 1+2q_2
\end{pmatrix}.
\end{equation}
Imposing the commutation condition
$\big[\partial^2_{qq}U,A\big]=0_{2\times2}$ (cf.\ Proposition~\ref{eq:commHess})
yields the constraint $2(a-d)q_1=0$, which implies $a=d$.
Hence, all compatible matrices are of the form
$A=aI_2+bJ_2$, where $J_2$ has zero diagonal entries and unit off-diagonal
entries.

The eigenvalues of $A$ are $\lambda_1=a+b$ and $\lambda_2=a-b$, with corresponding
orthonormal eigenvectors
$v_1=\tfrac{1}{\sqrt2}(1,1)$ and $v_2=\tfrac{1}{\sqrt2}(1,-1)$.
Let
\begin{equation}
Q=\frac1{\sqrt2}
\begin{pmatrix}
1 & 1\\
1 & -1
\end{pmatrix},
\end{equation}
be the orthogonal matrix collecting these eigenvectors. Then
$Q^T A Q=\mathrm{diag}(a+b,a-b)$, and the associated spectral coordinates
(Definition~\ref{def:spectral}) are
$x=Q^T q \in \mathbb{R}^2$, namely
$x_1=\frac{q_1+q_2}{\sqrt2}$ and $x_2=\frac{q_1-q_2}{\sqrt2}$.

In these coordinates the potential becomes
\begin{equation}
U(x_1,x_2)=\tfrac12(x_1^2+x_2^2)+\tfrac16 x_1^3-\tfrac16 x_2^3
=f(x_1)+g(x_2),
\end{equation}
with $f(x)=\frac12 x^2+\frac16 x^3$ and
$g(x)=\frac12 x^2-\frac16 x^3$.
This realizes the full spectral separation predicted by
Theorem~\ref{th:strong}, since the eigenvalues of $A$ are simple for $b\neq0$.
Accordingly, the equations of motion decouple into two independent nonlinear
oscillators,
\begin{equation}
\ddot x_1+x_1+\tfrac12 x_1^2=0,
\qquad
\ddot x_2+x_2-\tfrac12 x_2^2=0,
\end{equation}
with first integrals
$I_1=\frac12\dot x_1^2+f(x_1)$ and
$I_2=\frac12\dot x_2^2+g(x_2)$.

The canonical energy of the original Lagrangian is
\begin{equation}
E=I_1+I_2
=\tfrac12(\dot x_1^2+\dot x_2^2)
+\tfrac12(x_1^2+x_2^2)
+\tfrac16(x_1^3-x_2^3),
\end{equation}
while the compatible quadratic Lagrangian yields the additional conserved
quantity
$\widetilde E=\lambda_1 I_1+\lambda_2 I_2
=aE+b(I_1-I_2)$, explicitly,
\begin{equation}
\widetilde E=
(a+b)\!\left(\tfrac12\dot x_1^2+\tfrac12 x_1^2+\tfrac16 x_1^3\right)
+(a-b)\!\left(\tfrac12\dot x_2^2+\tfrac12 x_2^2-\tfrac16 x_2^3\right).
\end{equation}

Thus, the Sawada--Kotera system is completely separable in the spectral
coordinates associated with $A$, and the integrals $E$ and $\widetilde E$
correspond to linear combinations of the separated energies. These results
recover exactly those obtained in~\cite{GorniScomparinZampieri2024}.


\subsection{The $n$-dimensional extension of H\'enon--Heiles case}
\label{sec:HHhigher}

We consider an $n$--dimensional generalization of the classical H\'enon--Heiles
model, obtained by replacing one oscillator direction with an isotropic
$(n-1)$--dimensional sector while preserving the cubic interaction structure of the potential \cite{HenonHeiles1964}.
The classical H\'enon--Heiles system, originally introduced as a model in
galactic dynamics to investigate the existence of additional integrals of
motion, plays a central role in the theory of nonlinear Hamiltonian systems \cite{BlaszakRauchWojciechowski1994}.

Let us write  the natural Lagrangian in $\mathbb{R}^n$
\begin{equation}
\mathcal{L}(q,\dot{q}) = \frac{1}{2} \sum_{i=1}^n \dot{q}_i^2 - U(q),
\end{equation}
with the generalized cubic potential
\begin{equation}
U(q) = \frac{\alpha}{2} \sum_{i=1}^{n-1} q_i^2 + \frac{1}{2} \beta q_n^2 + a\, q_n \sum_{i=1}^{n-1} q_i^2 + b\, q_n^3.
\end{equation}
where $\alpha,\beta,a,b\in \mathbb{R}$ are constant parameters. Here, the classical H\'enon--Heiles system is recovered as the special case $n=2$.

The Hessian of the potential $U$ in block-form can be written as:
\begin{equation}
\label{eq:HessHH}
\partial^2_{qq}U = 
\begin{pmatrix}
(\alpha + 2 a q_n) I_{n-1} & 2 a \, q_{1:n-1} \\
2 a \, q_{1:n-1}^T & \beta + 6 b q_n
\end{pmatrix},
\end{equation}
where we have defined $q_{1:n-1} = (q_1, \dots, q_{n-1})^T\in \mathbb{R}^{n-1}$. Hence, following Proposition \ref{eq:commHess}, we look for a symmetric matrix
\begin{equation}
\label{eq:A_HH}
A = \begin{pmatrix} B & v \\ v^T & d \end{pmatrix},
\end{equation}
with $B \in \mathrm{Sym}(n,\mathbb{R})$, $v \in \mathbb{R}^{n-1}$ and $d \in \mathbb{R}$,
such that the commutator $\big[\partial^2_{qq}U,A\big]=0$. 
Writing down the commutator in blocks, we get the block-equation
\begin{equation}
\label{eq:commHH}
\begin{pmatrix}
\big[B, (\alpha + 2 a q_n) I_{n-1}\big] & B(2 a q_{1:n-1}) - v (\beta + 6 b q_n) \\
2 a q_{1:n-1}^T B - 2 a d q_{1:n-1}^T & v^T 2 a q_{1:n-1} - 2 a q^T_{1:n-1} v
\end{pmatrix} = 0_{n\times n}.
\end{equation}
We analyze each block of Eq.~\eqref{eq:commHH} separately, which must vanish identically. Our goal is to determine the matrices $B$ and $v$, and the scalar $d$, thus providing a possible parametrization of the matrix $A$ defined in Eq.~\eqref{eq:A_HH}. In particular,

\begin{itemize}
\item \textit{Upper-left block}.  
Let us compute the commutator
$[B,(\alpha+2a q_n)I_{n-1}] = B(\alpha+2a q_n)I_{n-1} - (\alpha+2a q_n)I_{n-1}B$.
Since the identity matrix $I_{n-1}$ commutes with every matrix, we have
$I_{n-1}B=BI_{n-1}$.
Hence, this block vanishes identically and does not impose any constraint on $B$.
\item \textit{Lower-left block}.  
Consider the condition $\mathbf{0}_{1\times(n-1)} = q_{1:n-1}^T B - d\,q_{1:n-1}^T$.
This can be rewritten as
$q_{1:n-1}^T(B-dI_{n-1})=\mathbf{0}_{1\times(n-1)}$.
Therefore, we conclude that $B=dI_{n-1}$.
\item \textit{Upper-right block}.  
Consider $B(2a q_{1:n-1}) - v(\beta + 6b q_n) = \mathbf{0}_{(n-1)\times1}$.
Substituting the expression for $B$ obtained in the previous block, we obtain
$2ad\,q_{1:n-1} - v(\beta + 6b q_n) = \mathbf{0}_{(n-1)\times1}$,
from which it follows that $v=\mu(q_n)\,q_{1:n-1}$, with
\begin{equation}
\label{eq:mu}
\mu(q_n)\equiv \frac{2ad}{\beta+6bq_n}\,.
\end{equation}
\item \textit{Lower-right block}.  
This block vanishes identically, since
$2a(v^T q_{1:n-1} - q_{1:n-1}^T v)
= 2a(v\cdot q_{1:n-1} - q_{1:n-1}\cdot v)=0$,
as the scalar product is symmetric.
\end{itemize}
In conclusion, the $A$ matrix can be rewritten as
\begin{equation}
A = 
\begin{pmatrix}
dI_{n-1} & \mu(q_n) q_{1:n-1} \\
\mu(q_n) q_{1:n-1}^T & d
\end{pmatrix}.
\end{equation}
\begin{proposition}
\label{prop:avlavtHH}
The spectrum of $A$ is given by
$\lambda_1=\dots=\lambda_{n-2}=d$,
and
$\lambda_\pm=d\pm\mu\,\|q_{1:n-1}\|$.
The eigenspace associated with the eigenvalue $d$ is
$\mathcal{E}_d=\{(y_{1:n-1},0)^T\in\mathbb{R}^n : y_{1:n-1}\cdot q_{1:n-1}=0\}$,
which has dimension $n-2$.
The eigenvectors corresponding to $\lambda_\pm$ are
$h_\pm=(q_{1:n-1},\pm\|q_{1:n-1}\|)^T$.
\end{proposition}
\begin{proof}
We recall that the determinant formula for block matrices \cite{Zhang2005} gives
\begin{equation}
\det
\begin{pmatrix}
P & R \\
S & T
\end{pmatrix}
=
\det\left(P\right)\det\left(T-SP^{-1}R\right),
\end{equation}
which is valid when $P$ is invertible. We take
$P=(d-\lambda)I_{n-1}$,
$R=\mu(q_n) q_{1:n-1}$,
$S=\mu(q_n) q_{1:n-1}^T$, and $T=d-\lambda$, where $\mu(q_n)$ is defined in Eq.~\eqref{eq:mu}.
Hence, since $P^{-1}=(d-\lambda)^{-1}I_{n-1}$, we obtain
\begin{equation}
\textstyle
\begin{aligned}
SP^{-1}R
&=\mu(q_n) q_{1:n-1}^T(d-\lambda)^{-1}I_{n-1}\mu(q_n)q_{1:n-1},\\
&=\mu^2(q_n) q_{1:n-1}^T q_{1:n-1}(d-\lambda)^{-1},\\
&=\mu^2(q_n)\|q_{1:n-1}\|^2(d-\lambda)^{-1}.
\end{aligned}
\end{equation}
The eigenvalues of $A$ are obtained from the vanishing of the characteristic polynomial. Using the above results, we get
\begin{equation}
\textstyle
\begin{aligned}
0
&\equiv\det(A-\lambda I_n),\\
&=\det \begin{pmatrix}
(d-\lambda)I_{n-1} & \mu(q_n) q_{1:n-1}\\
\mu(q_n) q_{1:n-1}^T & d-\lambda
\end{pmatrix},\\
&=\det\Big[(d-\lambda)I_{n-1}\Big]
\det\Big[(d-\lambda)-\mu^2(q_n)\|q_{1:n-1}\|^2(d-\lambda)^{-1}\Big],\\
&=(d-\lambda)^{n-1}
\Big[(d-\lambda)-\mu^2(q_n)\|q_{1:n-1}\|^2(d-\lambda)^{-1}\Big],\\
&=(d-\lambda)^{n-2}
\Big[(d-\lambda)^2-\mu^2(q_n)\|q_{1:n-1}\|^2\Big].
\end{aligned}
\end{equation}
We thus obtain the eigenvalues $\lambda_1=\dots=\lambda_{n-2}=d$
and $(d-\lambda_\pm)^2=\mu^2(q_n)\|q_{1:n-1}\|^2$.
Hence
$d-\lambda_\mp=\pm\mu(q_n)\|q_{1:n-1}\|$
and therefore
$\lambda_\pm=d\pm\mu(q_n)\|q_{1:n-1}\|$.

Let us now turn to the eigenvectors.
Let $h\equiv(y_{1:n-1},\eta)^T$ with $y_{1:n-1}\in\mathbb{R}^{n-1}$ and $\eta\in\mathbb{R}$.
The eigenvalue equation $(A-\lambda I_n)h=0_n$ yields
\begin{equation}
\begin{pmatrix}
(d-\lambda)I_{n-1} & \mu(q_n) q_{1:n-1} \\
\mu(q_n) q_{1:n-1}^T & d-\lambda
\end{pmatrix}
\begin{pmatrix}
y_{1:n-1}\\
\eta
\end{pmatrix}
=0,
\end{equation}
which gives the system
\begin{equation}
\begin{cases}
(d-\lambda)y_{1:n-1}+\mu(q_n) q_{1:n-1}\,\eta=0_{n-1}, \\
\mu(q_n) q_{1:n-1}^T y_{1:n-1}+(d-\lambda)\eta=0.
\end{cases}
\end{equation}
In particular,
\begin{itemize}
\item \textit{Case $\lambda_1=\dots=\lambda_{n-2}=d$.}
Setting $\lambda=d$ in the first equation yields
$\mu(q_n) q_{1:n-1}\,\eta=0$.
For $q_{1:n-1}\neq0_{n-1}$ this implies $\eta=0$.
The second equation then gives
$q_{1:n-1}^T y_{1:n-1}=0$.
Hence, the eigenvectors are
$h=(y_{1:n-1},0)$ such that $y_{1:n-1}\cdot q_{1:n-1}=0$,
which form an $(n-2)$-dimensional eigenspace.

\item \textit{Case $\lambda_\pm$.}
For $\lambda_\pm$ we have $d-\lambda_\pm=\mp\mu(q_n)\|q_{1:n-1}\|$.
From the first equation we obtain
$\mu(q_n) q_{1:n-1}\,\eta
=-(d-\lambda_\pm)y^\pm_{1:n-1}
=\pm\mu(q_n)\|q_{1:n-1}\|\,y^\pm_{1:n-1}$,
so that
$y^\pm_{1:n-1}=\pm(\eta/\|q_{1:n-1}\|)\,q_{1:n-1}$.
Substituting this expression into the second equation yields
$\pm\eta\mu(q_n)\bigl[q_{1:n-1}^T q_{1:n-1}/\|q_{1:n-1}\|
-\|q_{1:n-1}\|\bigr]=0$,
which vanishes identically since
$q_{1:n-1}^T q_{1:n-1}=\|q_{1:n-1}\|^2$.
Choosing $\eta=\|q_{1:n-1}\|$, we obtain
$y^\pm_{1:n-1}=\pm q_{1:n-1}$ and
$h_\pm=(q_{1:n-1},\pm\|q_{1:n-1}\|)$.
\end{itemize}
This concludes the demonstration.
\end{proof}
Let us define an orthonormal decomposition of $\mathbb{R}^n$ according to the eigenspaces of $A$:
$\mathbb{R}^n = \mathcal{E}_d \oplus \mathcal{E}_\pm$,
where $\mathcal{E}_d$ is the $(n-2)$-dimensional degenerate eigenspace associated with
$\lambda_1=\dots=\lambda_{n-2}=d$, and $\mathcal{E}_\pm$ is the two-dimensional plane
associated with the non-degenerate eigenvalues
$\lambda_\pm = d \pm \mu(q_n)\|q_{1:n-1}\|$.
Following Proposition~\ref{prop:avlavtHH}, we introduce an orthonormal basis of
$\mathcal{E}_d$:
$\mathscr{B}_{\mathcal{E}_d} = \bigl\{ (e_{(i)},0)^T \in \mathbb{R}^n :
e_{(i)} \in \mathbb{R}^{n-1},\ \|e_{(i)}\|=1,\ e_{(i)}\cdot q_{1:n-1}=0,\
i=1,\dots,n-2 \bigr\}$,
consisting of vectors orthogonal to $q_{1:n-1}$ and with zero component along
$q_n$.
The orthonormal basis of $\mathcal{E}_\pm$ is given by
$\mathscr{B}_{\mathcal{E}_\pm} =
\left\{
u_{(\pm)} = \frac{(q_{1:n-1}, \pm \|q_{1:n-1}\|)^T}
{\sqrt{2}\,\|q_{1:n-1}\|}
\right\}$,
where each $u_{(\pm)}$ is normalized and spans the plane associated with the
non-degenerate eigenvalues.
Let
\begin{equation}
Q = \bigl[\, e_{(1)}, \dots, e_{(n-2)}, u_{(+)}, u_{(-)} \,\bigr] \in O(n),
\end{equation}
be the orthogonal matrix whose columns are the vectors of the bases defined
above. Then
$Q^T A Q = \mathrm{diag}\bigl(d,\dots,d,\lambda_+,\lambda_-\bigr)$,
and the spectral coordinates are defined as
$x = Q^T q = (x_1, \dots, x_{n-2}, x_+, x_-)^T \in \mathbb{R}^n$,
with
\begin{equation}
\begin{cases}
x_i = e_{(i)}\cdot q, \qquad i=1,\dots,n-2,\\
x_\pm = u_{(\pm)} \cdot q.
\end{cases}
\end{equation}
In these coordinates, the $n$-dimensional potential splits as follows:
\begin{equation}
U(x) =
\frac{\alpha}{2} \sum_{i=1}^{n-2} x_i^2
+ \frac{\beta}{2} \left( x_+^2 + x_-^2 \right)
+ a (x_+ + x_-) \sum_{i=1}^{n-2} x_i^2
+ b (x_+ + x_-)^3.
\end{equation}
The final form of the equations of motion reads
\begin{equation}
\displaystyle
\begin{cases}
\displaystyle
\ddot{x}_i + \alpha x_i + 2 a x_i (x_+ + x_-) = 0,\\
\displaystyle
\ddot{x}_+ + \beta x_+ + a \sum_{i=1}^{n-2} x_i^2 + 3 b \,(x_+ + x_-)^2 = 0,\qquad i=1,\dots,n-2\\
\displaystyle
\ddot{x}_- + \beta x_- + a \sum_{i=1}^{n-2} x_i^2 + 3 b \,(x_+ + x_-)^2 = 0.
\end{cases}
\end{equation}
In this setting, the energy integrals are the canonical energy:
\begin{equation}
\begin{aligned}
E ={}& \frac{1}{2} \sum_{i=1}^{n-2} \dot{x}_i^2
+ \frac{1}{2} (\dot{x}_+^2 + \dot{x}_-^2)
+ \frac{\alpha}{2} \sum_{i=1}^{n-2} x_i^2 +\\
&+ \frac{\beta}{2} (x_+^2 + x_-^2)
+ a (x_+ + x_-) \sum_{i=1}^{n-2} x_i^2
+ b (x_+ + x_-)^3,
\end{aligned}
\end{equation}
and the compatible integral:
\begin{equation}
\begin{aligned}
\tilde{E} ={}& \frac{d}{2} \sum_{i=1}^{n-2}  \dot{x}_i^2
+ \frac{\lambda_+}{2} \dot{x}_+^2
+ \frac{\lambda_-}{2} \dot{x}_-^2
+ \frac{d\alpha}{2} \sum_{i=1}^{n-2} x_i^2 + \frac{\lambda_+ \beta}{2} x_+^2+\\
&
+ \frac{\lambda_- \beta}{2} x_-^2
+ d a \sum_{i=1}^{n-2} x_i^2 (x_+ + x_-)
+ \frac{b}{2}(\lambda_+ + \lambda_-) (x_+ + x_-)^3.
\end{aligned}
\end{equation}

Finally, note that the variables $x_i$ associated with the degenerate block
remain coupled to $x_\pm$, and therefore no separation occurs in the generic
case. This is consistent with Theorem~\ref{th:weak}, which (due to the presence
of degenerate eigenvalues in $\mathcal{E}_d$) requires the matrix $A$ to have constant
coefficients, whereas in the present case $\mu(q_n)$ is not constant.
\subsubsection{Spectral Separability and Parameter Restrictions}
In the general case, the compatible matrix $A$ defined in Eq.~\eqref{eq:A_HH}
depends on both the vector $q_{1:n-1}$ and the coordinate $q_n$ through the
off-diagonal term $\mu(q_n) q_{1:n-1}$.
Since $q_{1:n-1}$ varies along the motion, the matrix $A$ is not constant in
general. Even if $\mu(q_n)$ were constant, the presence of the varying vector $q_{1:n-1}$ would
still prevent our framework. Hence, in this case our theorems
require very restrictive conditions on the coefficients of $A$.

For $A$ to remain constant along the motion, it is necessary that
(i) $q_{1:n-1}$ be constant and (ii) $\mu(q_n)$ be constant.
Condition (i) can be satisfied only if $a=0$, since this eliminates the
$q_{1:n-1}$-dependent terms in the Hessian \eqref{eq:HessHH}, i.e., the cubic
interaction term $q_n \sum_{i=1}^{n-1} q_i^2$ must be absent.
Condition (ii) implies either $b=0$ or $a=0$.

Therefore, the matrix $A$ is constant only in the following cases:

\begin{itemize}
\item ($a\neq 0$, $b=0$). The cubic term $q_n^3$ vanishes, $A$ can be made constant
along appropriate directions, and the system separates into $n-2$ linear
harmonic oscillators plus a two-dimensional integrable H\'enon--Heiles
subsystem, corresponding to the classical cases $1{:}6{:}1$, $1{:}6{:}8$,
$1{:}12{:}16$, etc.
\item ($a=0$, $b\neq 0$). The potential reduces to a sum of independent
harmonic oscillators; $A$ is trivial and the system is separable, but the
H\'enon--Heiles subsystem is absent.
\end{itemize}

In all other cases ($a\neq 0$, $b\neq 0$), both $q_{1:n-1}$ and $\mu(q_n)$ vary
along the motion, the matrix $A$ is non-constant, and spectral separation is
impossible.

\begin{table}[h]
\centering
\label{tab:separation-cases}
\begin{tabular}{|c|l|l|}
\hline
Case & $(a,b)$-parameters & Separation \\ \hline
{\small I} & {\small $a\neq 0$, $b=0$} & {\small HH integrable $+$ $n-2$ harmonic oscillators} \\ \hline
{\small II} & {\small $a=0$, $b\neq 0$} & {\small Harmonic oscillators only, trivial} \\ \hline
{\small III} & {\small $a\neq 0$, $b\neq 0$} & {\small Non-separable, $x_i$ coupled to $x_\pm$} \\ \hline
\end{tabular}
\caption{\small Classification of the system according to the potential parameters
$(a,b)$ and the corresponding separability properties.}
\end{table}

This analysis shows that the requirement that the matrix $A$ be constant
coincides precisely with the classical integrable H\'enon--Heiles cases,
whereas any other choice of parameters leads, within the criteria of our
theory, to a general, non-separable system.


\subsection{An example in $\mathbb{R}^4$ with a transcendental potential}

We illustrate the previous construction on a nontrivial example in $\mathbb{R}^4$ involving a transcendental potential. We consider the natural Lagrangian
\begin{equation}
\mathcal{L}(q,\dot q) = \frac12 \sum_{i=1}^4 \dot q_i^2 - U(q_1,q_2,q_3,q_4),
\end{equation}
where
\begin{equation}
U(q_1,q_2,q_3,q_4)
= e^{q_1}\sin(q_2) + \arctan(q_3) + \ln(1+q_4^2) + q_1 q_3 + q_2 q_4.
\end{equation}

Again, the key requirement is the compatibility relation
$\big[\partial^2_{qq}U,A\big]=0_{4\times4}$.
The structure of $U$ suggests a natural pairing of variables $(q_1,q_3)$ and $(q_2,q_4)$, since the potential contains linear couplings $q_1 q_3$ and $q_2 q_4$. For this reason, we assume a block form for $A$ of the type
\begin{equation}
A =
\begin{pmatrix}
a & 0 & b & 0 \\
0 & d & 0 & e \\
b & 0 & c & 0 \\
0 & e & 0 & f
\end{pmatrix}.
\end{equation}

To ensure that a scalar potential $\widetilde U$ exists, we impose the integrability conditions corresponding to the symmetry of mixed partial derivatives. This yields the constraints
$
a - c = b$ and  $d - f = e.
$
Among the admissible choices, we select the simplest symmetric solution
$a=c=1$, $b=1$, $d=f=1$, $e=1$
which leads to
\begin{equation}
A =
\begin{pmatrix}
1 & 0 & 1 & 0 \\
0 & 1 & 0 & 1 \\
1 & 0 & 1 & 0 \\
0 & 1 & 0 & 1
\end{pmatrix}.
\end{equation}

The next step is to understand the dynamical structure induced by $A$. Each $2\times 2$ block has eigenvalues $2$ and $0$, which suggests introducing coordinates adapted to its eigenspaces. The corresponding orthogonal change of variables is given by
\begin{equation}
x_1 = \frac{q_1+q_3}{\sqrt{2}}, \quad
x_2 = \frac{q_1-q_3}{\sqrt{2}}, \quad
x_3 = \frac{q_2+q_4}{\sqrt{2}}, \quad
x_4 = \frac{q_2-q_4}{\sqrt{2}}.
\end{equation}
In these coordinates, the kinetic matrix becomes diagonal:
\begin{equation}
Q^T A Q = \mathrm{diag}(2,0,2,0),
\end{equation}
We now express the potential in the new coordinates. Substituting the inverse transformation
we obtain a natural decomposition into the $(x_1,x_2)$ and $(x_3,x_4)$ blocks, which share the same eigenvalue.
\begin{equation}
U(q) = H_1(x_1,x_2) + H_2(x_3,x_4),
\end{equation}
where each function depends only on one eigenspace of $A$. Explicitly,
\begin{equation}
H_1(x_1,x_2)
= e^{(x_1+x_2)/\sqrt{2}}
+ \arctan\!\left(\frac{x_1-x_2}{\sqrt{2}}\right)
+ \frac{x_1^2 - x_2^2}{2},
\end{equation}
\begin{equation}
H_2(x_3,x_4)
= \sin\!\left(\frac{x_3+x_4}{\sqrt{2}}\right)
+ \ln\!\left[1+\left(\frac{x_3-x_4}{\sqrt{2}}\right)^2\right]
+ \frac{x_3^2 - x_4^2}{2}.
\end{equation}
This decomposition has a direct dynamical consequence: in the diagonal coordinates, the Euler--Lagrange equations split into two independent subsystems,
\begin{equation}
2 \ddot x_1 + \partial_{x_1}H_1 = 0, \quad
\partial_{x_2} H_1 = 0,\quad
2 \ddot x_3 + \partial_{x_3} H_2 = 0, \quad
\partial_{x_4} H_2 = 0.
\end{equation}
Thus, the full dynamics decomposes into two weakly coupled pairs $(x_1,x_2)$ and $(x_3,x_4)$, each evolving independently of the other.
Finally, each subsystem admits a conserved energy,
\begin{equation}
E_1 = \frac12 \dot x_1^2 + H_1(x_1,x_2), \qquad
E_2 = \frac12 \dot x_3^2 + H_2(x_3,x_4),
\end{equation}
so that the total energy splits as $E = E_1 + E_2$.

This example shows that even highly non-polynomial potentials can be incorporated into the theory.


\section{Addendum: inverse characterization of separable systems with non-symmetric kinetic matrices}
\label{sec:inv}

Let us consider a quadratic Lagrangian function and the associated
Euler--Lagrange equations
\begin{equation}
\label{eq:Lagasym}
\widetilde{\mathcal L}(q,\dot q)
=\frac{1}{2}\sum_{i,j=1}^n A_{ij}\dot q_i\dot q_j-\widetilde U(q),
\qquad
\sum_{j=1}^n A_{ij}\ddot q_j+\partial_{q_i}\widetilde U(q)=0,
\end{equation} 
where $i=1,\dots,n$, $A\in M_n(\mathbb{R})$ is a constant invertible matrix
(not necessarily symmetric), and
$\widetilde U\in C^2(\mathbb{R}^n)$.
The following theorem investigates which class of canonical Lagrangians
generates the same Euler--Lagrange equations in separated form.
\begin{proposition}[Inverse theorem for a non-symmetric kinetic matrix]
Consider the Lagrangian \eqref{eq:Lagasym}.
If there exists a set of univariate functions
$f_i(q_i)\in C^2(\mathbb{R})$ such that
\begin{equation}
\label{eq:condE}
\partial_{q_i}\widetilde U(q)
=
\sum_{j=1}^n A_{ij} f_j'(q_j),
\qquad i=1,\dots,n,
\end{equation}
then there exists a separated Lagrangian
$\mathcal L(q,\dot q)=\frac{1}{2}\sum_{i=1}^n \dot q_i^2-\sum_{i=1}^n f_i(q_i),$
which generates the same Euler--Lagrange equations as
$\widetilde{\mathcal L}$, namely
\begin{equation}
\label{eq:aaa}
\ddot q_i+f_i'(q_i)=0,
\qquad i=1,\dots,n.
\end{equation}
\end{proposition}

\begin{proof}
The Euler--Lagrange equations associated with $\mathcal L$ are
$\ddot q_i+f_i'(q_i)=0$, $i=1,\dots,n$, and are completely separated.
Consider the Euler--Lagrange equations \eqref{eq:Lagasym} for
$\widetilde{\mathcal L}$ and substitute condition~\eqref{eq:condE} into them.
It then follows that the equations can be written as
$
\sum_{j=1}^n A_{ij}\bigl[\ddot q_j+f_j'(q_j)\bigr]=0.
$
Since $A$ is invertible by hypothesis, $\mathcal L$ and
$\widetilde{\mathcal L}$ generate the same separated Euler--Lagrange equations.

Conversely, if the separated system
$\ddot q_i+f_i'(q_i)=0$ generates the same equations as the Lagrangian
$\widetilde{\mathcal L}$, comparison of the equations necessarily yields
condition~\eqref{eq:condE}. This proves the equivalence and completes the proof.
\end{proof}
The next step is to determine the form of the potential $\widetilde U$.
\begin{proposition}
Assume the compatibility condition $A_{ij}\alpha_j = A_{ji}\alpha_i$.
Then the only potential $\widetilde U$ compatible with \eqref{eq:condE} is
\begin{equation}
\label{eq:potentialU}
\widetilde U(q)=
\sum_{i,j=1}^n \big(\alpha_j A_{ij} q_i q_j
+ \beta_j A_{ij} q_i\big) + \mathrm{const},
\end{equation}
where $\alpha, \beta \in \mathbb{R}^n$. Moreover, the corresponding
Euler--Lagrange equations \eqref{eq:aaa} take the form
\begin{equation}
\ddot q_i + 2\alpha_i q_i + \beta_i = 0,
\qquad i = 1, \dots, n.
\end{equation}
\end{proposition}
\begin{proof}
Condition \eqref{eq:condE} can be written in vector form as
$
\nabla \widetilde U(q) = A F(q),
$ with
$
F(q) = (f_1'(q_1), \dots, f_n'(q_n))^T.
$
For such a potential $\widetilde U$ to exist, the vector field $A F(q)$ must be conservative.
Hence, the integrability condition
$
\partial_{q_k} [ \sum_j A_{ij} f_j'(q_j) ]
=
\partial_{q_i} [ \sum_j A_{kj} f_j'(q_j)]
$
must hold for all $i,k$.
Since each $f_j'$ depends only on $q_j$, for $i \neq k$ this reduces to
$
A_{ik} f_k''(q_k) = A_{ki} f_i''(q_i).
$
The left-hand side depends only on $q_k$, while the right-hand side depends only on $q_i$; therefore, both must be constant. It follows that $f_i''$ is constant for every $i$, and thus
$
f_i'(q_i) = 2\alpha_i q_i + \beta_i.
$
Integrating, we obtain
$
f_i(q_i) = \alpha_i q_i^2 + \beta_i q_i + \mathrm{const}.
$
Substituting into \eqref{eq:condE} and integrating once more yields \eqref{eq:potentialU}.
\end{proof}

This proposition provides a complete inverse characterization, identifying precisely when a system with non-symmetric kinetical couplings is dynamically equivalent to a separated one.
An immediate consequence is that the only admissible separated potentials are quadratic. 
Therefore, the class of separable non-symmetric Lagrangians corresponds to systems whose effective dynamics reduces to a collection of independent harmonic oscillators after a suitable transformation.


\nocite{*}
\bibliographystyle{rendiconti}
\bibliography{bibliography}

@article{Stackel1893,
  author  = {St{\"a}ckel, Paul},
  title   = {{\"U}ber die Bewegung eines Punktes in einer n-fachen Mannigfaltigkeit},
  journal = {Mathematische Annalen},
  volume  = {42},
  pages   = {537--563},
  year    = {1893}
}

@article{LeviCivita1904,
  author  = {Levi-Civita, Tullio},
  title   = {Sulla integrazione della equazione di {Hamilton--Jacobi} per separazione di variabili},
  journal = {Mathematische Annalen},
  volume  = {59},
  pages   = {383--397},
  year    = {1904}
}

@article{Eisenhart1934,
  author  = {Eisenhart, Luther Pfahler},
  title   = {Separable systems of {St{\"a}ckel}},
  journal = {Annals of Mathematics},
  volume  = {35},
  number  = {2},
  pages   = {284--305},
  year    = {1934}
}

@article{Woodhouse1975,
  author  = {Woodhouse, N. M. J.},
  title   = {Killing tensors and the separation of the {Hamilton--Jacobi} equation},
  journal = {Communications in Mathematical Physics},
  volume  = {44},
  pages   = {9--38},
  year    = {1975}
}

@article{KalninsMiller1980,
  author  = {Kalnins, E. G. and Miller, Willard},
  title   = {Killing tensors and variable separation for {Hamilton--Jacobi} and {Helmholtz} equations},
  journal = {SIAM Journal on Mathematical Analysis},
  volume  = {11},
  pages   = {1011--1026},
  year    = {1980}
}

@article{Benenti1997,
  author  = {Benenti, Sergio},
  title   = {Intrinsic characterization of the variable separation in the {Hamilton--Jacobi} equation},
  journal = {Journal of Mathematical Physics},
  volume  = {38},
  number  = {12},
  pages   = {6578--6602},
  year    = {1997}
}

@article{BenentiChanuRastelli2001,
  author  = {Benenti, Sergio and Chanu, Claudia and Rastelli, Giovanni},
  title   = {Variable separation for natural Hamiltonians with scalar and vector potentials on Riemannian manifolds},
  journal = {Journal of Mathematical Physics},
  volume  = {42},
  number  = {5},
  pages   = {2065--2091},
  year    = {2001}
}

@article{BenentiChanuRastelli2005,
  author  = {Benenti, Sergio and Chanu, Claudia and Rastelli, Giovanni},
  title   = {Variable-separation theory for the null {Hamilton--Jacobi} equation},
  journal = {Journal of Mathematical Physics},
  volume  = {46},
  pages   = {042901},
  year    = {2005}
}

@article{SarletRamos2000,
  author  = {Sarlet, Willy and Ramos, Ant{\'o}nio},
  title   = {Adjoint symmetries, separability, and volume forms},
  journal = {Journal of Mathematical Physics},
  volume  = {41},
  pages   = {2877--2888},
  year    = {2000}
}

@article{PrinceSarletThompson2002,
  author  = {Prince, G. E. and Sarlet, Willy and Thompson, G.},
  title   = {The inverse problem of the calculus of variations: The use of geometrical calculus in Douglas's analysis},
  journal = {Transactions of the American Mathematical Society},
  volume  = {354},
  number  = {7},
  pages   = {2897--2919},
  year    = {2002}
}

@article{TopalovMatveev2003,
  author  = {Topalov, Peter and Matveev, Vladimir S.},
  title   = {Geodesic equivalence via integrability},
  journal = {Geometriae Dedicata},
  volume  = {96},
  pages   = {91--115},
  year    = {2003}
}

@article{SawadaKotera1974,
  author  = {Sawada, K. and Kotera, T.},
  title   = {A Method for Finding $N$-Soliton Solutions of the KdV Equation and KdV-Like Equation},
  journal = {Progress of Theoretical Physics},
  volume  = {51},
  pages   = {1355--1367},
  year    = {1974}
}

@article{HenonHeiles1964,
  author  = {H{\'e}non, M. and Heiles, C.},
  title   = {The Applicability of the Third Integral of Motion},
  journal = {Astronomical Journal},
  volume  = {69},
  pages   = {73--79},
  year    = {1964}
}

@article{GorniScomparinZampieri2024,
  author  = {Gorni, Gianluca and Scomparin, Mattia and Zampieri, Gaetano},
  title   = {A New Class of Separable Lagrangian Systems Generalizing the Sawada--Kotera System},
  journal = {Dynamics},
  volume  = {4},
  number  = {3},
  pages   = {499--505},
  year    = {2024},
  doi     = {10.3390/dynamics4030026}
}

@article{BlaszakRauchWojciechowski1994,
  author  = {B{\l}aszak, M. and Rauch-Wojciechowski, S.},
  title   = {A generalized H{\'e}non--Heiles system and related integrable Newton equations},
  journal = {Journal of Mathematical Physics},
  volume  = {35},
  pages   = {1693--1709},
  year    = {1994},
  doi     = {10.1063/1.530565}
}

@book{Zhang2005,
  author    = {Zhang, Fuzhen},
  title     = {The Schur Complement and Its Applications},
  series    = {Numerical Methods and Algorithms},
  volume    = {4},
  publisher = {Springer},
  address   = {New York},
  year      = {2005},
  doi       = {10.1007/b105056}
}

@article{FuchssteinerOevel1982,
  author    = {Benno Fuchssteiner and Walter Oevel},
  title     = {The bi-Hamiltonian structure of some nonlinear fifth- and seventh-order differential equations and recursion formulas for their symmetries and conserved covariants},
  journal   = {Journal of Mathematical Physics},
  volume    = {23},
  number    = {3},
  pages     = {358--363},
  year      = {1982},
  doi       = {10.1063/1.525376}
}

@article{Fordy1983,
  author    = {A. P. Fordy and B. Dorizzi and B. Grammaticos},
  title     = {Painlev\'e problem for coupled ordinary differential equations},
  journal   = {Physical Review Letters},
  volume    = {50},
  number    = {6},
  pages     = {374--377},
  year      = {1983},
  doi       = {10.1103/PhysRevLett.50.374}
}

@article{CosgroveScoufis1993,
  author    = {C. M. Cosgrove and G. Scoufis},
  title     = {Painlev\'e classification of a class of differential equations of the second order and second degree},
  journal   = {Studies in Applied Mathematics},
  volume    = {88},
  number    = {1},
  pages     = {25--87},
  year      = {1993},
  doi       = {10.1002/sapm199388125}
}


\end{document}